\documentclass[12pt,onecolumn]{IEEEtran}

\IEEEoverridecommandlockouts
\usepackage{amsmath,amssymb,amsthm,mathrsfs}
\usepackage{graphicx}
\usepackage{cite}
\usepackage{flushend}
\usepackage{url}

\allowdisplaybreaks

\makeatletter%
\if@twocolumn%
\newcommand{\Figwidth}{\columnwidth}%
\newcommand{\Tablewidth}{0.42\textwidth}%
\def\twocolbreak{\nonumber\\ &}%
\else
\newcommand{\Figwidth}{4.0in}%
\newcommand{\Tablewidth}{0.6\textwidth}%
\def\twocolbreak{}%
\fi%
\makeatother%

\begin{document}

\title{On the Separability of Ergodic Fading MIMO Channels: A Lattice Coding Approach}
\author{Ahmed Hindy and Aria Nosratinia
\thanks{The authors are with the department of Electrical Engineering, University of Texas at Dallas, Email: ahmed.hindy@utdallas.edu and aria@utdallas.edu}
\thanks{This work was supported in part by the grant ECCS1546969 from the National Science Foundation.}
}

\maketitle



\newtheorem{theorem}{Theorem}
\newtheorem{lemma}{Lemma}
\newtheorem{remark}{Remark}
\newtheorem{corollary}{Corollary}
\newtheorem{definition}{Definition}

\def\Px{\rho}
\def\Ph{\sigma_h^2}
\def\Ix{\boldsymbol{I}}
\def\hv{\boldsymbol{h}}
\def\xv{\boldsymbol{x}}
\def\yv{\boldsymbol{y}}
\def\wv{\boldsymbol{w}}
\def\ev{\boldsymbol{e}}
\def\tv{\boldsymbol{t}}
\def\dv{\boldsymbol{d}}
\def\zv{\boldsymbol{z}}
\def\vv{\boldsymbol{v}}
\def\sv{\boldsymbol{s}}
\def\gv{\boldsymbol{g}}
\def\qv{\boldsymbol{q}}
\def\latpoint{\boldsymbol{\lambda}}
\def\SNRm{\boldsymbol{\Psi}}
\def\Hm{\boldsymbol{H}}
\def\Hall{\boldsymbol{H}_d}
\def\Am{\boldsymbol{A}}
\def\Bm{\boldsymbol{B}}
\def\Um{\boldsymbol{U}}
\def\Lm{\boldsymbol{L}}
\def\Fm{\boldsymbol{F}}
\def\Gm{\boldsymbol{G}}
\def\Vm{\boldsymbol{V}}
\def\Dm{\boldsymbol{D}}
\def\Wm{\boldsymbol{W}}
\def\comp{\mathbb{C}}
\def\real{\mathbb{R}}
\def\ints{\mathbb{Z}}
\def\lat{\Lambda}
\def\voronoi{\mathcal{V}}
\def\genvoronoi{\Omega}
\def\Ex{\mathbb{E}}
\def\vol{\text{Vol}}
\def\gap{\Delta}
\def\ball{\mathcal{B}}
\def\cov{\Sigma}
\def\covx{\boldsymbol{K_x}}
\def\Cs{\mathcal{L}}
\def\order{\boldsymbol{\pi}}
\def\impulse{\boldsymbol{\delta}}
\def\ss{\mathcal{H}}
\def\hh{\mathfrak{h}}
\def\prob{\mathbb{P}}
\def\bigO{\mathcal{O}}


\begin{abstract}
This paper addresses point-to-point communication over block-fading channels with independent fading blocks. When both channel state information at the transmitter (CSIT) and receiver (CSIR) are available, most achievable schemes use  separable coding, i.e., coding independently and in parallel over different fading states. Unfortunately, separable coding has drawbacks including large memory requirements at both communication ends. In this paper a lattice coding and decoding scheme is proposed that achieves the ergodic capacity without separable coding, with lattice codebooks and decoding decision regions that are universal across channel realizations. We first demonstrate this result for fading distributions with discrete, finite support whose sequences are {\em robustly typical}. Results are then extended to continuous fading distributions, as well as multiple-input multiple-output (MIMO) systems. In addition, a variant of the proposed scheme is presented for the MIMO ergodic fading channel with CSIR only, where we prove the existence of a universal codebook that achieves rates within a constant gap to capacity for finite-support fading distributions. The gap is small compared with other schemes in the literature. Extension to continuous-valued fading is also provided.
\end{abstract}

\begin{IEEEkeywords}
Ergodic capacity, lattice codes, separable coding, MIMO.
\end{IEEEkeywords}


\section{Introduction}
\label{intro}

For the band-limited Additive White Gaussian Noise (AWGN) channel, approaching capacity with manageable complexity has been extensively studied~\cite{TCM,MLC,coset,LDPC,Turbo,BICM,polar,Forney_Gaussian}. McEliece and Stark~\cite{stark_paper} established the ergodic capacity of the Gaussian fading channel with CSIR only. Goldsmith and Varaiya~\cite{ptp_ergodic} extended the result for full CSI (both CSIT and CSIR).  
The capacity of the ergodic fading MIMO channel with isotropic fading and CSIR was established by Telatar~\cite{telatar} and Foschini and Gans~\cite{foschini}. For a survey of related results please see Biglieri \textit{et al.}~\cite{fading_Shamai}.

Under fading and in the presence of CSIT, one straightforward capacity approaching technique is separable coding, i.e., coding independently and in parallel over different fading states of the channel~\cite{ptp_ergodic,capacity_MIMO}. Unfortunately, in practice separable coding imposes heavy costs that are further magnified in the presence of low-probability fading states. In particular either rate loss due to discarding low probability fading states, or loss of coding performance due to shorter block lengths, must be tolerated. In addition, separable coding requires operating multiple encoders and decoders with different transmission rates in parallel, which requires large memory at both communication ends.
Thereby, achieving the ergodic capacity of block-fading channels {\em without} separable coding remains an important and interesting question.%
\footnote{It was pointed out in~\cite{fading_Shamai} that under maximum likelihood decoding the ergodic capacity of point-to-point channels with CSIT can be attained using Gaussian signaling {\em without} separable coding. However, one cannot directly conclude that the same result holds for non-Gaussian (structured) codebooks.} 

This paper shows that non-separable lattice coding and decoding achieve the ergodic capacity of the block fading SISO channel. At the transmitter, the symbols of the codeword are permuted across time. Time-varying Minimum Mean-Square Error (MMSE) scaling is used at the receiver, followed by a decoder that is universal for all fading realizations drawn from a given fading distribution. Thus, the codebook and decision regions are fixed across transmissions; the only channel-dependent blocks are the permutation and the MMSE scaling. 
We first highlight the main ideas of the proposed scheme in the context of a heuristic channel model that motivates the proposed approach. We then generalize the solution to all fading distributions whose realizations are robustly typical, and to continuous distributions via a bounding argument.
The results are then extended to MIMO block-fading channels.

 A lattice coding and decoding scheme is also proposed for the ergodic fading MIMO channel with CSIR only, where the channel coefficients are drawn from a discrete distribution with finite support.  In this setting, channel-matching decision regions are proposed, where we use a worst case error bounding technique to show the existence of a universal lattice codebook that achieves rates within a constant gap to capacity for all fading realizations. The gap is infinitesimal in some special cases. We also extend the scheme to continuous-valued fading, and show that the rates achieved are close to capacity under Rayleigh fading.
%

Lattice coding has an extensive literature. De~Buda addressed the optimality of lattice codes for the AWGN channel~\cite{debuda_capacity}, a result later corrected by Linder \textit{et al.}~\cite{Linder}. Loeliger~\cite{Loeliger} proved the achievability of $\frac{1}{2} \log(\text{SNR})$ with lattice coding and decoding. Urbanke and Rimoldi~\cite{Urbanke_lattices} showed the achievability of $\frac{1}{2} \log(1+\text{SNR})$ with maximum likelihood decoding. Erez and Zamir~\cite{Erez_Zamir} showed that lattice coding and decoding achieve the capacity of the AWGN  channel, where the ingredients of the achievable scheme  include nested lattice codes in addition to common randomness via a dither variable and MMSE scaling at the receiver.  Erez \textit{et al.}~\cite{Lattices_good} also proved the existence of lattices with good properties that achieve the performance promised in~\cite{Erez_Zamir}. 
El~Gamal \textit{et al.}~\cite{DMT_lattices} showed that nested lattice codes achieve the white-input capacity, as well as the optimal diversity-multiplexing tradeoff, of the AWGN MIMO channel with fixed channel coefficients. Recently, Zhan \textit{et al.}~\cite{IF_RX} proposed a novel technique that is based on nested lattice codes together with \textit{integer-forcing linear receivers}, where the receiver decodes integer combinations of the signals at each antenna, similar to the~\textit{compute-and-forward} technique~\cite{CF_Nazer}.  
Ordentlich and Erez showed that in conjunction with a precoder that is independent of the channel, integer-forcing can operate within a constant gap to the MIMO capacity~\cite{IF_Erez}. 
In~\cite[Section~4.5]{dispersion} Vituri analyzed the performance of lattice codes under fading channels without power constraint.
Under ergodic fading and CSIR only, Luzzi and Vehkalahti~\cite{Luzzi_j} recently showed that a class of lattices belonging to a family of division algebra codes achieve rates within a constant gap to capacity, however, this gap can be large. In~\cite{Journal1}, a lattice coding scheme was proposed whose decoder does not depend on the fading realizations, achieving rates within a constant gap to capacity. The results in both~\cite{Luzzi_j,Journal1} are limited to channels with isotropic fading, i.e., the optimal input covariance matrix is a scaled identity. Lately, Liu and Ling~\cite{polar_lattices} showed that polar lattices achieve the capacity of the SISO i.i.d.\ fading channel. Campello \textit{et al.}~\cite{algebraic_lattices} also proved that lattices constructed from algebraic codes achieve the SISO ergodic capacity. Unfortunately neither~\cite{polar_lattices} nor~\cite{algebraic_lattices} are easily extendable to MIMO channels.

The remainder of the paper is organized as
follows. Section~\ref{sec:signal} establishes the notation and provides
an overview of lattices and typicality. Section~\ref{sec:CSIT} presents the
lattice coding scheme under full CSI, and Section~\ref{sec:CSIR} under
CSIR only.  Section~\ref{sec:conc} provides a concluding summary.


\section{Preliminaries}
\label{sec:signal}

\subsection{Notation and Definitions} 
\label{sec:notation}

Throughout the paper we use the following notation. Boldface lowercase letters denote column vectors and boldface uppercase letters denote matrices. The sets of real numbers and integers are denoted by $\real,\ints$, respectively. $\boldsymbol{A}^T$ denotes the transpose of matrix~$\boldsymbol{A}$. $a_i$ is element~$i$ of~$\boldsymbol{a}$. 
$\det(\Am)$ and $\text{tr}(\Am)$ denote the determinant and trace of the square matrix~$\Am$, respectively. $\Ix_n$ is the size-$n$ identity matrix. $\boldsymbol{0_n}$ and $\boldsymbol{1_n}$ denote the all-zero and all-one $n \times n$ matrices, respectively. $\prob,\Ex$ denote probability and expectation, respectively, and~$\prob_e$ represents error probability. $\ball_n(q)$ is an $n$-dimensional ball of radius~$q$ and the volume of shape $\mathcal{A}$ is $\vol(\mathcal{A})$. $\kappa^+ \triangleq \max \{ \kappa,0 \}$. $|\mathcal{A}|$ denotes the number of elements in set~$\mathcal{A}$. Unless otherwise specified, all logarithms are in base~$2$.


\subsection{Lattice Codes} 
\label{sec:lattice}

A lattice $\lat$ is a discrete subgroup of $\real^n$ which is closed under reflection and real addition. 
The fundamental Voronoi region~$\voronoi$ of the lattice~$\lat$ is defined by
\begin{equation}
\voronoi = \big \{ \boldsymbol{s} \in \real^n: \text{arg} \min_{\latpoint \in \lat} || \boldsymbol{s} - \latpoint ||  = \boldsymbol{0} \big \}. 
\label{fund_voronoi}
\end{equation}
The {\em second moment per dimension} of  $\lat$ is defined as 
\begin{equation}
\sigma_{\lat}^2 = \frac{1}{n \vol(\voronoi)} \int_{\voronoi} ||\boldsymbol{s}||^2 d \boldsymbol{s}.
\label{moment}
\end{equation}
Every $\boldsymbol{s} \in \real^n$ can be uniquely written as $\boldsymbol{s}= \latpoint + \ev$ where $\latpoint \in \lat$, $\ev \in \voronoi$, with ties broken in a systematic manner. The quantizer is then defined by
\begin{equation}
Q_{\voronoi}(\boldsymbol{s})= \latpoint \;, \quad \text{if } \boldsymbol{s} \in \latpoint+\voronoi. 
\label{quantizer}
\end{equation}
Define the modulo-$\lat$ operation corresponding to $\voronoi$ as follows
\begin{equation}
[\boldsymbol{s}] \, \text{mod} \lat \triangleq \boldsymbol{s}- Q_{\voronoi}(\boldsymbol{s}).
\label{mod}
\end{equation}
The modulo-$\lat$ operation also satisfies
\begin{equation}
\big[ \boldsymbol{s} + \boldsymbol{t} \big] \, \text{mod} \lat = 
\big[ \boldsymbol{s} + [\boldsymbol{t}] \, \text{mod} \lat \big] \, \text{mod} \lat
\hspace{5mm} \forall \boldsymbol{s},\boldsymbol{t} \in \real^n.
\label{mod_mod}
\end{equation}
The lattice $\lat$ is nested in~$\lat_1$ if~$\lat \subseteq \lat_1$. 
We employ the class of nested lattice codes proposed in~\cite{Erez_Zamir}. 
For completeness, the lattice construction is outlined as follows:
\begin{enumerate}
\item Draw an i.i.d. vector~$\gv \triangleq [g_1, \ldots, g_n]^T$ whose elements are uniformly distributed on the set $\{ 0, 1, \ldots, q-1 \}$, where $q$ is a large prime number.
\item Define the codebook $\mathcal{C} = \big \{ \zv \in \ints^n : \zv = \big[ \gv \beta \big] \, \text{mod} \, q , \quad \beta = 0,\ldots,q-1$ \big \}. 
\item Apply Construction~$A$ to lift $\mathcal{C}$ to $\real^n$ such that $\lat'_1 = q^{-1} \mathcal{C} + \ints^n$.
\end{enumerate}
A self-similar pair of nested lattices is used such that the coarse lattice $\lat = \eta \lat'_1$, where the scaling factor~$\eta$ assures $\lat$ has second moment~$\Px$, and $\lat_1 = \frac{1}{\tau} \lat'_1 = \frac{1}{\eta \, \tau} \lat$ where $\tau$ scales the fundamental volume of~$\lat_1$ to achieve rate~$R$ given by 
\begin{equation}
R \, \triangleq \, \frac{1}{n} \log \frac{\vol(\voronoi)}{\vol(\voronoi_1) } \,,
\label{lattice_rate}
\end{equation} 
and $\voronoi,\voronoi_1$ are the Voronoi regions of the coarse and fine lattices, respectively. 
The ensemble of nested lattice pairs employed above have been shown to be simultaneously good for AWGN coding, packing, covering and quantization~\cite{Lattices_good}.
The covering goodness of $\lat$ is defined by 
\begin{equation}
\lim_{n \to \infty} \frac{1}{n} \log \frac{\vol(\ball_n(R_c))}{\vol(\ball_n(R_f))}=0 \, ,
\label{covering}
\end{equation}
where the covering radius $R_c$ is the radius of the smallest sphere spanning $\voronoi$ and ${R_f}$ is the radius of the sphere whose volume is equal to $\vol(\voronoi)$. 
\begin{definition}\cite[Theorem~1]{DMT_lattices}
Let $\mathit{f}: \real^n \to \real$ be a Riemann integrable function of bounded support (i.e., $\mathit{f} (\zv) = 0$ if $|\zv|$ exceeds some bound). 
An ensemble of lattices~$\{ \lat \}$ with fundamental volume~$\vol(\voronoi)$ satisfies the Minkowski-Hlawka Theorem if for any $\epsilon>0$ there exists a lattice with dimension~$n$ such that
\begin{equation}
\bigg |  \Ex_{\lat} \Big [ \sum_{\zv \in \lat, \zv \neq \boldsymbol{0}} \mathit{f}(\zv) \Big] \, -  \, 
\frac{1}{\vol(\voronoi)}  \int_{\real^n}  \mathit{f} (\zv) d \zv  \bigg | \,  < \, \epsilon \, . 
\label{eq:MHS}
\end{equation}
\label{def:MHS}
\end{definition}
 
\begin{lemma}\cite[Theorem~2]{DMT_lattices}
The ensemble of nested lattice pairs of~\cite{Erez_Zamir} satisfies the Minkowski-Hlawka Theorem at large dimension~$n$. 
\label{lemma:MHS}
\end{lemma}

A key ingredient of the lattice coding scheme proposed in~\cite{Erez_Zamir} is using common randomness (dither)~$\dv$, drawn uniformly over $\voronoi$, in conjunction with the lattice code. The following lemma from \cite{Erez_Zamir} is key to the development of the results in this paper. 

\begin{lemma}\cite[Lemma~1]{Erez_Zamir}
For any point $\tv \in \voronoi$ that is independent of a dither $\dv$ drawn uniformly over a lattice Voronoi region~$\voronoi$, the point $\xv \triangleq \big[\tv - \dv \big] \, \text{mod} \lat$ is uniformly distributed over $\voronoi$ and is also independent of~$\tv$.
\label{lemma:uniform}
\end{lemma}


\subsection{Typicality} 
\label{sec:typicality}

We briefly review robust typicality~\cite[Appendix]{robust_typicality} and weak typicality~\cite[Chapter~3.1]{Cover}. Consider a probability distribution $\prob$ on the discrete domain ${\cal A}=\{  \alpha_1, \alpha_2, \ldots, \alpha_{\chi} \}$. 
\begin{definition}
A $\delta$-robustly typical set $T_{\delta}^{(R)}$ according to $\prob$  is the set of all sequences~$\xv \in {\cal A}^n$ that satisfy
\begin{equation}
|n_k - n \prob_k| \leq \delta n \prob_k \, ,
\label{robust_1}
\end{equation}
for all $k \in \{1,\ldots, \chi \}$, where $\prob_k$ stands for $\prob (\alpha_k)$ and $n_k$ the number of coordinates of $\xv$ that are equal to~$\alpha_k$. 
\label{def:typicality}
\end{definition}
Long random sequences drawn i.i.d. are with high probability robustly typical according to the underlying distribution, as indicated by the following result.
\begin{lemma} \cite[Lemma 17]{robust_typicality}
The probability of a sequence~$\xv$ of length~$n$ {\em not} being $\delta$-robustly typical is upper bounded by
\begin{align}
\prob (\xv \notin T_{\delta}^{(R)}) & \leq \sum_{k=1}^{\chi} \prob \big( |n_k - n \prob_k| > \delta n \prob_k \big)
\twocolbreak
\leq  2 \chi e^{-\delta^2 \mu n/3},
\label{robust_2}
\end{align}
where $\mu \triangleq \min \prob_k$ is the smallest non-zero probability in~$\prob$. 
\label{lemma:typicality}
\end{lemma}

Weak typicality~\cite{Cover} is defined here via entropy rates.
\begin{definition}
An $\epsilon$-weakly typical set $T_{\epsilon}^{(W)}$ with respect to a sequence of probability distributions
$\prob \big ( x_1,\ldots,x_n \big)$ is defined as the set of all vectors $\xv=[x_1,\ldots,x_n]$ that satisfy
\begin{equation}
2^{-n \big( \hbar(x) + \epsilon  \big)} \leq \prob (\xv) \leq 2^{-n \big( \hbar(x) - \epsilon  \big)} ,
\label{robust_w1}
\end{equation}
where~$\hbar(x) \triangleq \lim \limits_{n \to \infty} \frac{1}{n} \sum_{\xv} -\prob(\xv) \log \prob(\xv)$ is the entropy rate of the sequence of probability distributions, assuming it exists.\footnote{A prominent example is when the sequence of probability laws is stationary.}
\label{def:typicality_weak}
\end{definition}
The probability of an arbitrary sequence of length~$n$ being weakly typical is $\prob \big(\xv \in T_{\epsilon}^{(W)} \big) > 1 - \epsilon$. The cardinality of $T_{\epsilon}^{(W)}$  is bounded by
\begin{equation}
|T_{\epsilon}^{(W)}| \, \leq \, 2^{n \big( \hbar(x) + \epsilon  \big)} .
\label{robust_w2}
\end{equation}


\section{A Capacity Achieving Lattice Coding Scheme}
\label{sec:CSIT}

Consider a real-valued single-antenna point-to-point channel with block-fading and i.i.d.\ Gaussian noise. 
The received signal is given by $y_i = h_i x_i + w_i$. The transmission and reception of a codeword over $n$ channel uses is represented by
\begin{equation}
\yv=\Hm \xv+\wv,
\label{sig_Rx}
\end{equation}
where $\Hm$ is an $n \times n$ diagonal matrix whose diagonal entries~$h_i$ are drawn from a discrete distribution with finite-support~$\ss$. The channel coherence length is $b$ with $n=n' b$, where $b$ is fixed and $n'$ is proportional to~$n$. Therefore each codeword experiences $n'$ independent fading realizations. The covariance of the channel is 
\begin{equation}
\boldsymbol{\cov_h} = \Ph
\begin{bmatrix}
 \boldsymbol{1_b} \hspace{4mm} \boldsymbol{0_b} \hspace{4mm} \ldots \hspace{4mm} \boldsymbol{0_b}  \\
\boldsymbol{0_b} \hspace{4mm} \boldsymbol{1_b} \hspace{4mm} \ldots \hspace{4mm} \boldsymbol{0_b}  \\
\ddots \\
\boldsymbol{0_b} \hspace{4mm} \boldsymbol{0_b} \hspace{4mm} \ldots \hspace{4mm} \boldsymbol{1_b}
\end{bmatrix}.
\label{C_H}
\end{equation}
 Both the transmitter and receiver have full knowledge of the channel state.
The noise~$\wv \in \real^n$ is zero-mean i.i.d. Gaussian with covariance~$\Ix_n$ and is independent of~$\Hm$. 
$\xv \in \real^n$ is the codeword, subject to an average power constraint~$\frac{1}{n} \Ex \big[ || \xv ||^2 \big]  \leq \Px$.

The ergodic capacity of the real-valued point-to-point channel is given by~\cite{ptp_ergodic}
\begin{equation}
C \, = \, \frac{1}{2} \Ex_h \big[ \log{\big( 1 +  h^2 \Px^*(h) \big)} \big],
\label{capacity_ptp}
\end{equation}
where $\Px^*(h)$ denotes the channel-dependent  waterfilling power allocation~\cite{ptp_ergodic}, which satisfies~$\Ex_h [ \Px^*(h) ]= \Px$.
 This capacity is achieved via separable coding~\cite{ptp_ergodic}, which is defined as follows
\begin{definition}
In a separable coding scheme, the ergodic fading channel over time is demultiplexed into virtual parallel channels according to fading states, over which independent codewords are transmitted. Each codeword is therefore transmitted over multiple occurrences of the same fading state.
\label{def:separability}
\end{definition}

To highlight the essential ideas of the proposed scheme we first address the problem in the context of a heuristic channel model. 


\subsection{The Random Location Channel}
\label{sec:heuristic}

We define a channel model, called the random location channel, where in each block of length~$n$, denoted $\hv \triangleq [h_1,\ldots, h_n]$, the empirical frequency of occurrence of each channel state perfectly matches the underlying probability distribution. 
Channel coefficients $h_i$ take values from the set $\ss \triangleq \big \{ \hh_1,\ldots,\hh_{|\ss|} \big \}$. 
Consider sequences $[h_1,\ldots,h_n]$ that satisfy $n_{\hh_k} = n \prob_{\hh_k}$. The random location channel draws from this set of sequences with equal probability. Thus, the channel is by construction perfectly robustly-typical. The transmitter knows non-causally the number of occurrences of each  $\hh_k$ in $\hv$, however, {\em their location is random,} and only known {\em causally} at both the transmitter and receiver (full CSI). The model provides a stepping stone for the achievable scheme proposed for the ergodic channel in Section~\ref{sec:general}, and serves to illustrate its underlying intuitions.

\begin{theorem}
For the random location channel defined above, the rate
\begin{equation}
R < \frac{1}{2} \sum_{s=1}^{|\ss|} \mu_s  \log{\big( 1 +  \hh_s^2 \Px^*(\hh_s) \big)} 
\end{equation}
is achievable using {\em non-separable} lattice coding, where~$\mu_s$ represents the frequency of occurrence of coefficient value~$\hh_s$ such that $\sum_{s=1}^{|\ss|} \mu_s = 1$, and~$\Px^*(\hh_s)$ is the waterfilling power allocation for channel coefficient~$\hh_s$ drawn from~$\ss$.
\label{theorem:rate_heuristic}
\end{theorem}

\begin{proof} 

\textit{Encoding:} Nested lattice codes are used where~$\lat \subseteq \lat_1$. The transmitter emits a lattice point~$\tv \in \lat_1$ that is dithered with~$\dv$ which is drawn uniformly over~$\voronoi$. The dithered codeword is as follows
\begin{equation}
\xv = \big[\tv - \dv \big]~\text{mod} \lat \,
= \, \tv - \dv + \latpoint,
\label{sig_tx}
\end{equation}
where $\latpoint= - Q_{\voronoi}(\tv - \dv ) \in \lat$ from~\eqref{mod}. The coarse lattice~$\lat \subset \real^{n}$ has a second moment~$\Px$. The codeword is then multiplied by two cascaded matrices as follows 
\begin{equation}
\xv' = \Dm \, \Vm \, \xv,
\label{sig_tx_2}
\end{equation}
where~$\Vm$ is a permutation matrix  and~$\Dm$ is a diagonal matrix with~$D_{ii} = \sqrt{\Px^*(h_i)/\Px}$, where~$\Px^*(h_i)$ is the optimal waterfilling power allocation for the fading coefficient~$|h_i|$, as given in~\cite{ptp_ergodic}. Hereafter we use~$\Px^*_i$ as a short-hand notation for~$\Px^*(h_i)$. 
We show in Appendix~\ref{appendix:power} that the average power constraint of~$\xv'$ is  approximately the same as~$\xv$. 

\textit{Decoding:} The received signal $\yv$~is multiplied by a matrix~$\Um~\in~\real^{n \times n}$ cascaded with an inverse permutation matrix~$\Vm^T$, and the dither is removed as follows
\begin{align}
\yv' =& \Vm^T \Um \yv + \dv  \nonumber \\
 =&  \xv + (\Vm^T \Um \Hm \Dm \Vm - \Ix_n) \xv + \Vm^T \Um \wv + \dv    \nonumber \\ 
 =&  \tv + \latpoint + (\Vm^T \Um \Hm \Dm \Vm - \Ix_n) \xv + \Vm^T \Um \wv,    \nonumber \\ 
 =&  \tv + \latpoint + \zv,
\label{sig_rx_2}
\end{align}
where
\begin{equation}
\zv \triangleq (\Vm^T \Um \Hm \Dm \Vm - \Ix_n) \xv  + \Vm^T \Um \wv,
\label{eq_noise}
\end{equation}
and~$\zv$ is independent of~$\tv$ from Lemma~\ref{lemma:uniform}. 

The receiver matrix~$\Um$ is chosen to be the MMSE matrix given by
\begin{equation}
\Um= \Px \Hm \Dm ( \Ix_{n} + \Px \Hm^2 \Dm^2 )^{-1}.
\label{eq:U_MSE}
\end{equation}
$\Um$ is diagonal, where $U_{ii}= \Px D_{ii} h_i/(1+ \Px D_{ii}^2 h_i^2)$.  Now, the diagonal elements of~$\Um$ are
\begin{equation}
U_{ii}= \frac {\sqrt{\Px \Px^*_i } h_i}{1+ \Px^*_i h_i^2}.
\label{eq:ui_MSE}
\end{equation}
With a slight abuse of notation, define the permutation function~$\order$ such that $(h_{\order(1)},h_{\order(2)},\ldots,h_{\order(n)})$ represent the channel coefficients arranged in ascending order of the magnitudes. Consider permutation matrix~$\Vm$ such that $\Hm_{\order} \triangleq \Vm^T \Hm \Vm$, where the diagonal entries of~$\Hm_{\order}$ are~$h_{\order(i)}$. See Appendix~\ref{appendix:V} for further details on~$\Vm$.
From~\eqref{eq_noise} and \eqref{eq:U_MSE}, $z_i$ are given by\,%
\footnote{Since waterfilling dedicates more power to channels with larger magnitude, $ h_i^2 \geq h_j^2$ implies $\Px^*_i h_i^2 \geq \Px^*_j h_j^2$~\cite{ptp_ergodic}.}
\begin{equation}
z_i= \frac{-1}{\Px^*_{\order(i)}  h_{\order(i)}^2 +1} x_i + \frac{\sqrt{\Px \Px^*_{\order(i)} } h_{\order(i)}}{\Px^*_{\order(i)} h_{\order(i)}^2 +1} w_{\order(i)}.
\label{eq:zi}
\end{equation}

The following lemma, whose proof can be found in Appendix~\ref{appendix:z_dist}, elaborates some geometric properties of~$\zv$.

\begin{lemma}
For any $\epsilon>0$ and $\gamma>0$, there exists $n_{\gamma,\epsilon}$ such that for all $n>n_{\gamma,\epsilon}$, %
\begin{equation}
   \prob \big( \zv \notin \genvoronoi \big) < \gamma,
\label{eq:error_event}
\end{equation}
where $\genvoronoi$ is an $n$-dimensional ellipsoid, given by
\begin{equation}
\genvoronoi \triangleq \{ \sv \in \real^n~:~ \sv^T \boldsymbol{\cov}^{-1} \sv \leq (1+\epsilon) n \},
\label{voronoi}
\end{equation}
and~$\boldsymbol{\cov}$ is a diagonal matrix whose diagonal elements are given by
\begin{equation}
\cov_{ii} =  \frac{\Px}{\Px^*_{\order(i)} h_{\order(i)}^2 + 1}
\label{z_cov}
\end{equation}
\label{lemma:z_dist}
\end{lemma}

Now, we apply a version of the ambiguity decoder proposed in~\cite{Loeliger}, defined by an ellipsoidal decision region $\genvoronoi$ in~\eqref{voronoi}.%
\footnote{$\genvoronoi$  is a bounded measurable region of~$\real^{n}$~\cite{Loeliger}.} 
 The decoder chooses $\hat{\tv} \in \lat_1$ if and only if the received point falls exclusively within the decision region of the lattice point~$\hat{\tv}$, i.e., $\yv' \in \hat{\tv} + \genvoronoi$.

\textit{Probability of error:} As shown in~\cite[Theorem~4]{Loeliger}, on averaging over the  ensemble of fine lattices~$\Cs$ of rate~$R$ whose construction follows Section~\ref{sec:lattice}, the probability of error can be bounded by
\begin{align}
\frac{1}{|\Cs|} \sum_{\Cs} \, \prob_e <& \prob (\zv \notin \genvoronoi) + (1+ \delta) \, \frac{\vol(\genvoronoi)}{\vol(\voronoi_1)} \nonumber \\
=& \prob (\zv \notin \genvoronoi) + (1+ \delta) 2^{nR} \, \frac{\vol(\genvoronoi)}{\vol(\voronoi)}, 
\label{error_prob}
\end{align}
for any $\delta > 0$, and the equality follows from~\eqref{lattice_rate}. This is a union bound involving two events: the event that~$\zv$ is outside the decision region, i.e., $\{\zv \notin \genvoronoi\}$ and the event that~$\zv$ is in the intersection of two decision regions $\big\{\yv' \in \{ \tv_1 + \genvoronoi \} \cap \{ \tv_2 + \genvoronoi \} \big\}$, where $\tv_1,~\tv_2 \in \lat_1$ are two distinct lattice points. 
From Lemma~\ref{lemma:z_dist}, the first term in~\eqref{error_prob} is bounded by~$\gamma$.  
Consequently, the error probability can be bounded by 
\begin{equation}
\frac{1}{|\Cs|} \sum_{\Cs} \, \prob_e < \gamma + (1+ \delta) 2^{nR} \frac{\vol(\genvoronoi)}{\vol(\voronoi)}, 
\label{error_prob_2}
\end{equation}
for any $\gamma,\delta>0$. The volume of~$\genvoronoi$ is given by
\begin{equation}
\vol(\genvoronoi)= (1+ \epsilon)^{\frac{n}{2}} 
\vol \big( \ball(\sqrt{n \Px}) \big)
\Big (\prod_{i=1}^{n} \frac{1}{\Px^*_i h_i^2 + 1} \Big ) ^{\frac{1}{2}}.
\label{vol_omega}
\end{equation}

The second term in~\eqref{error_prob_2} is then bounded by
\begin{align}
& \quad (1+ \delta) 2^{nR} (1+ \epsilon)^{n/2} \Big (\prod_{i=1}^{n} \frac{1}{\Px^*_i h_i^2 + 1} \Big ) ^{\frac{1}{2}} \frac{\vol(\ball(\sqrt{n \Px}))}{\vol(\voronoi)} \nonumber \\
&=(1+ \delta) 2^{ -n \Big ( - \frac{1}{n} \log \big( \frac{\vol(\ball(\sqrt{n \Px}))}{\vol(\voronoi)} \big) + \xi \Big ) }, \label{eq:exponential}
\end{align}
where 
\begin{align}
\xi \triangleq& \frac{-1}{2} \log({1+ \epsilon}) 
 - \frac{1}{2n} \log \Big (\prod_{i=1}^{n} \frac{1}{\Px^*_i h_i^2 + 1} \Big ) - R \nonumber \\
= & \frac{-1}{2} \log({1+ \epsilon})  + \frac{1}{2n} \sum_{i=1}^{n} \log ( 1 + \Px^*_i h_i^2 ) - R.  \nonumber \\
= & \frac{-1}{2} \log({1+ \epsilon})  + \frac{1}{2} \sum_{s=1}^{|\ss|} \mu_s \log ( 1 + \Px^*_s \hh_s^2 ) - R \, ,  
\label{xi}
\end{align}
and~\eqref{xi} follows from the structure of the random location channel. From~\eqref{covering}, since the lattice $\lat$ is good for covering, the first term of the exponent in~\eqref{eq:exponential} vanishes. 
From~\eqref{eq:exponential}, whenever $\xi$ is a positive constant we have 
 $\prob_e \to 0$ as $n \to \infty$. Hence, positive $\xi$ can be achieved as long as
\begin{equation}
R < \, \frac{1}{2} \sum_{s=1}^{|\ss|} \mu_s \log ( 1 + \Px^*_s \hh_s^2 )  - \frac{1}{2} \log({1+ \epsilon}) - \epsilon',
\end{equation} 
where~$\epsilon$, $\epsilon'$ diminish with~$n$. 
The existence of a fine lattice that achieves the probability of error averaged over the ensemble of lattices~$\Cs$ is straightforward. 
The outcome of the decoding process is the lattice point $\hat{\tv}$, where in the event of successful decoding the noise is eliminated and from~\eqref{sig_rx_2}, $\hat{\tv}= \tv + \latpoint$. On applying the modulo-$\lat$ operation on~$\hat{\tv}$,
\begin{equation}
[\hat{\tv}]\text{ mod}\lat \, = \, [\tv+\latpoint]\text{ mod}\lat \, = \, \tv,
\label{mod_lambda}
\end{equation}
where the second equality follows from~\eqref{mod_mod} since~$\latpoint \in \lat$. 
Following in the footsteps of~\cite{DMT_lattices}, it can be shown that the error probability of the {\em Euclidean lattice decoder} is upper bounded by the error probability of the ellipsoidal decision region in~\eqref{voronoi}. The Euclidean lattice decoder is given by
\begin{equation}
\hat{\tv} = \text{arg} \min_{\tv' \in \lat_1} || \cov^{\frac{-1}{2}} ( \yv' - \tv' ) ||^2,
\label{Euclid}
\end{equation}
followed by the modulo-$\lat$ operation in~\eqref{mod_lambda}. This concludes the proof of Theorem~\ref{theorem:rate_heuristic}. 
\end{proof}


\subsection{Ergodic Fading}
\label{sec:general}

Now, we are ready to address the ergodic fading channel whose channel coefficients are drawn from a discrete distribution with finite support. Unlike the random location channel discussed earlier, in the following the number of occurrences of~$\alpha_k$ within a block is no longer fixed. 
\begin{theorem}
Non-separable lattice coding achieves the ergodic capacity of block-fading channels whose channel coefficients are drawn from an arbitrary discrete distribution with finite-support, when channel state information is available at all nodes. 
\label{theorem:rate_ptp}
\end{theorem}

\begin{proof}
The proof appears in Appendix~\ref{appendix:theorem_ptp}; here we provide a sketch. We follow a {\em best effort} approach in designing the permutation matrix~$\Vm$. In order to account for the ordering errors, we use a fixed decision region $\tilde{\genvoronoi}_1$ that is slightly larger than  $\genvoronoi^{(p)}$ (the decision region resulting from perfect channel ordering, which is non-realizable due to the causality of the channel knowledge). However, when the channel is robustly typical, the total number of ordering errors is negligible at large~$n$, and hence  the rate loss incurred by using larger decision regions vanishes.  
\end{proof}

The extension of Theorem~\ref{theorem:rate_ptp} to complex-valued channels is straightforward, using techniques similar to~\cite[Theorem 2]{paper1}. The channel would then be ordered with respect to the magnitude of channel coefficients.


\subsection{Extension to Continuous-Valued Fading}
\label{sec:continuous1}

In order to extend the arguments to continuous-valued fading channels, we assume the fading distribution possesses a finite second moment. We note that with full CSI, the information density contributed by each transmission is a strictly increasing function of the absolute value of the fading coefficient. First, let~$\tilde{g} \triangleq |h|^2 \Px_h / \Px$ denote the squared channel gain times the normalized waterfilling power allocation for that channel gain. Thus, we can partition the continuous values~$\tilde{g}$ into $L+1$~sets $G_\ell \triangleq [g_{\ell-1},g_{\ell}]$, where  $g_0\triangleq 0$ and $g_{L+1}=\infty$. For any sequence of channel gains~$\tilde{g}$ drawn from a continuous distribution, we quantize~$\tilde{g}$ to the lower limit of the bracket $G_i$ to which it belongs, producing a discrete random variable~$g$ taking values over the set $\{g_0,g_1,\ldots,g_L\}$. Note that the independence of the continuous-valued fading realizations guarantees the independence of the discrete-valued counterparts, and hence robust typicality would still apply.
We show that the rate $R$ supported by the discrete-valued channel $g$ is within a gap to capacity that can be bounded as follows
\begin{align}
C-R = & \, \Ex [\log (1+ \Px \tilde{g}) ] \, - \,  \Ex [\log (1+ \Px g) ]   \nonumber \\
= & \, \,  \Ex [\log (\frac{1+ \Px \tilde{g}}{1+\Px g}) | \tilde{g} \leq g_L] \prob (\tilde{g} \leq g_L) 
\twocolbreak 
+ \Ex [\log (\frac{1+ \Px \tilde{g}}{1+\Px g_L}) | \tilde{g} > g_L] \prob (\tilde{g} > g_L)   \label{eqn_gap_quant} \\
< & \,  \max \Big\{\log(\frac{1+\Px g_{\ell}}{1+\Px g_{\ell-1}})  \Big\}_{\ell=1}^L 
\twocolbreak 
+ \Ex [\log (\frac{1+ \Px \tilde{g}}{1+\Px g_L}) | \tilde{g} > g_L] \prob (\tilde{g} > g_L)  \nonumber \\
< & \,  \max \Big\{\log \big(1 + \Px (g_{\ell}-g_{\ell-1}) \big)  \Big\}_{\ell=1}^L
\twocolbreak 
+  \Ex [\log (1 + \frac{\Px (\tilde{g}-g_L)}{1+\Px g_L}) | \tilde{g} > g_L] \prob (\tilde{g} > g_L)  \nonumber \\
< & \, \gamma_1 + \Ex [\log (1 + \frac{ \tilde{g}-g_L}{g_L}) | \tilde{g} > g_L] \prob (\tilde{g} > g_L)  \nonumber \\
= & \, \gamma_1 + \Ex [\log (\frac{\tilde{g}}{g_L}) | \tilde{g} > g_L] \prob (\tilde{g} > g_L)  \nonumber \\
< & \, \gamma_1 + c \, ( \frac{\Ex [\tilde{g} | \tilde{g} > g_L]}{g_L} - 1) \prob (\tilde{g} > g_L)  
\label{eqn_gap_log} \\
< & \, \gamma_1 + c \, ( \frac{\Ex [\tilde{g}]}{g_L \prob (\tilde{g} > g_L)} - 1) \prob (\tilde{g} > g_L)  \label{eqn_gap_expectation} \\
< & \, \gamma_1 + \frac{c  \, \Ex [\tilde{g}]}{g_L} \, \triangleq \, \gamma_1 + \gamma_2 \, ,
\label{eqn_gap_continuous}
\end{align}
where $c \triangleq \log e$, and $\gamma_1 \triangleq \max \Big\{\log \big(1 + \Px (g_{\ell}-g_{\ell-1}) \big)  \Big\}_{\ell=1}^L$. \eqref{eqn_gap_log} follows since $\log_e(x) < x-1$ for all $x>0$ and \eqref{eqn_gap_expectation} follows from the law of total expectation.  $\gamma_1$ vanishes when $\max \{ g_{\ell} - g_{\ell-1} \}_{i=1}^L \ll \frac{1}{\Px}$, while $\gamma_2$ vanishes when  $g_L \gg  \Ex [\tilde{g}]$. 
Note that a necessary condition for $\gamma_2$ to vanish is that~$\Ex [g]$ is finite.

The gap is bounded more tightly when the distribution of~$\tilde{g}$ has a vanishing tail. For instance, when~$\tilde{g}$ is exponential, 
\begin{align}
C-R < & \, \gamma_1 + c \, ( \frac{\Ex [\tilde{g} | \tilde{g} > g_L]}{g_L} - 1) \prob (\tilde{g} > g_L)  \nonumber \\
< & \, \gamma_1 + c \, ( \frac{\Ex [\tilde{g} + g_L]}{g_L} - 1) \prob (\tilde{g} > g_L)  \label{eqn_gap_memoryless} \\
< & \, \gamma_1 + \frac{c \, \Ex [\tilde{g}]}{g_L} \, e^{-\frac{g_L}{\Ex[\tilde{g}]}}, 
\label{eqn_gap_Rayleigh}
\end{align}
which vanishes exponentially with~$g_L$. \eqref{eqn_gap_memoryless} follows since~$\tilde{g}$ is exponentially distributed and hence memoryless. 

To summarize, the gap bounding argument can be described as follows: Given $L+1$ channel quantization bins, we bound the total rate loss due to quantization by the rate loss in each of the bins. The first $L$ terms bound the amount of loss in rate by the input-output information density at the highest versus the lowest channel gain in each bracket $G_1,\ldots,G_L$. This strategy will not work for the final bin because the channel gain in $G_{L+1}$ is unbounded, instead we use the total rate contributed by the bin $G_{L+1}$ as a bound. Fortunately, this term also vanishes at large~$g_L$ since the probability of occurrence of such fading values is small enough.


\subsection{Extension to MIMO}
\label{sec:MIMO}

The result in Theorem~\ref{theorem:rate_ptp} can be extended to an $M \times N$ MIMO channel with full CSI. The received signal at time~$i$ is given by 
\begin{equation}
\yv_i=\Hm_i \xv_i+\wv_i,
\label{sig_Rx_MIMO}
\end{equation}
where $\Hm_i \in \real^{N \times M}$ is the channel-coefficient matrix.
 
\begin{theorem}
Lattice codes achieve the ergodic capacity of the MIMO block fading channel with channel state information available at both transmitter and receiver. This result holds for both discrete-valued and continuous-valued channels.   
\label{theorem:rate_MIMO}
\end{theorem}

\begin{proof}
Since~$\Hm_i$ are known perfectly, the transmitter and receiver can transform the MIMO channel into $\mathcal{S} \triangleq \min \{ M, N \}$ SISO parallel channels via singular-value decomposition. The SISO individual capacities can be achieved as shown in Section~\ref{sec:general}. Let the singular-value decomposition of~$\Hm_i$ be $\Hm_i = \Bm_i \Lm_i \Fm_i^T $, where $\Bm_i \in \real^{N \times N} $,  $\Fm_i \in \real^{M \times M} $  are orthonormal matrices representing the left and right eigenvalue matrices of~$\Hm_i$, respectively. $\Lm_i$ is an $N \times M$ rectangular diagonal matrix with $\mathcal{S}$ non-zero values on the main diagonal. Hence, at the receiver, the received signal is spatially equalized as follows
\begin{equation}
\tilde{\yv}_i = \Bm_i^T \yv_i,
\label{spatial_decoding}
\end{equation}
and at the transmitter, the signal is spatially precoded such that
\begin{equation}
\xv_i = \Fm_i \tilde{\xv}_i.
\label{spatial_precoding}
\end{equation}
From~\eqref{sig_Rx_MIMO}--\eqref{spatial_precoding}, $\tilde{\yv}_i$ can be represented by
\begin{equation}
\tilde{\yv}_i = \Lm_i \tilde{\xv}_i + \tilde{\wv_i},
\label{sig_Rx_parallel}
\end{equation}
where $\tilde{\wv_i} \triangleq \Bm_i^T \wv_i \in \real^N$ is i.i.d. Gaussian, since $\Bm_i$ is orthonormal. Each element in $\tilde{\yv}_i \in \real^N$ is then 
\begin{equation}
\tilde{y}_i^{(\iota)} = \ell_i^{(\iota)} \tilde{x}_i^{(\iota)} + \tilde{w}_i^{(\iota)} , \qquad \iota = 1, \ldots, \mathcal{S}, 
\label{sig_Rx_SISO}
\end{equation}
where $\ell_i^{(1)},\ell_i^{(2)}, \ldots, \ell_i^{(\mathcal{S})}$ represent the singular values of~$\Hm_i$ in descending order. The received signal in~\eqref{sig_Rx_SISO} is equivalent to a set of~$\mathcal{S}$ parallel channels, whose individual capacities can be achieved similar to Section~\ref{sec:general} via transmitting~$\mathcal{S}$ simultaneous lattice codebooks across antennas. The final step would be allocating the optimal power policy, which is waterfilling over time and space, as follows~\cite[Section~8.2.3]{Tse}. Assuming that the joint probability distribution of $\ell^{(1)}, \ldots, \ell^{(\mathcal{S})}$ is known, the power of stream~$\iota$ at time~$i$ is given by
\begin{equation}
P_i^{(\iota)} =  \{ c - \frac{1}{(\ell_i^{(\iota)})^2} \}^+, 
\label{waterfilling_1}
\end{equation}
where $c$ is chosen such that
\begin{equation}
c \triangleq \sum_{\iota=1}^\mathcal{S} \Ex \big[ \{ c - \frac{1}{(\ell_i^{(\iota)})^2} \}^+  \big] = P,
\label{waterfilling_2}
\end{equation}
and~$P$ is the average power constraint. 
The extension to continuous-valued channels is similar to SISO and is omitted. This concludes the proof of Theorem~\ref{theorem:rate_MIMO}.
\end{proof}


\section{The MIMO Channel Without CSIT}
\label{sec:CSIR}

In this section we consider the $M \times N$ MIMO point-to-point channel with CSIR only. The received signal at time~$i$ is given by $\yv_i=\Hm_i \xv_i+\wv_i$, where $\Hm_i \in \real^{N \times M}$ is the channel-coefficient matrix at time~$i$. For convenience channels gains are taken to be real-valued; the extension to complex-valued channels is straight forward and similar to~\cite{paper1}.
The channel experiences block fading
with coherence length $b$, thus $\Hm_i$ are identically distributed, and any two of them are independent if and only if taken from different fading blocks. For convenience, we also define $\Hm$ to obey the same distribution, standing in for the prototypical MIMO channel gain matrix without reference to a specific time.  Each codeword consists of $n$ channel uses, where $n$ is an integer multiple of the fading block length, i.e., $n=n' b$. 
The transmitter knows the channel distribution, including the coherence length, but not the channel realizations. 
$\xv_i \in \real^M$ is the transmitted vector at time~$i$, where the codeword
\begin{equation}
\xv \triangleq [\xv_1^T \, \xv_2^T , \ldots \xv_n^T]^T
\label{codeword_MIMO}
\end{equation} 
is transmitted throughout~$n$ channel uses and satisfies $\Ex [ || \xv ||^2 ] \leq n \Px$. Unlike the achievable scheme in Section~\ref{sec:general}, each codeword is transmitted across both space and time. The noise~$\wv \in \real^{N n}$ defined by $\wv \triangleq [\wv_1^T , \ldots , \wv_n^T]^T$ is zero-mean i.i.d. Gaussian with covariance $ \Ix_{N n} $. For convenience we define the SNR per transmit antenna to be $\Px' \triangleq \Px/ M$.
The ergodic capacity of the MIMO channel is given by~\cite{Tse}
\begin{equation}
C =  \max_{\text{tr}(\covx) \leq \Px }  \frac{1}{2} \,  \Ex_H \big[ \log  \det ( \Ix_N +  \Hm \covx \Hm^T ) \big],
\label{capacity_MIMO_gen}
\end{equation}
where~$\covx$ is the covariance matrix of each super-symbol~$\xv_i$.
For a sequence of channel coefficients~$\{\Hm_i\}_{i=1}^n$ drawn from an underlying distribution, weak law of large numbers implies that for each positive $\eta$ and $\tilde{\epsilon}$, a finite $n$ exists such that
\begin{align}
\prob \bigg( \Big | & \frac{1}{2n} \sum_{i=1}^n\log \det \big( \Ix_M + \Px' \Hm_i^T \Hm_i \big) \,
\twocolbreak
 - \, \frac{1}{2} \Ex \big[ \log \det ( \Ix_M + \Px' \Hm^T \Hm) \big] \Big | \, \geq  \, \eta \bigg) < \tilde{\epsilon} \, , 
\label{eq:rateMean}
\end{align}
Hence, the expression $\frac{1}{2n} \sum_{i=1}^n\log \det \big( \Ix_M + \Px' \Hm_i^T \Hm_i \big)$ approaches its statistical mean with high probability as $n$~grows. Hereafter we denote the left-hand side probability in~\eqref{eq:rateMean} by~$\prob_{\text{out}}^{(\eta)}$.

\begin{lemma}
 Consider a MIMO channel~$ \yv = \Hm_s \xv + \wv $, where $\Hm_s \triangleq \text{diag} \big( \Hm_1, \ldots, \Hm_n  \big)$, and $\Hm_1, \ldots, \Hm_n$ are realizations of a stationary and ergodic process, and are only known at the receiver. Then there exists at least one lattice codebook that achieves rates satisfying
 \begin{equation}
 R < \frac{1}{2} \Ex \big[ \log \det ( \Ix_M + \Px' \Hm^T \Hm) \big] \, - \, \eta \, ,
 \label{capacity_eqn}
 \end{equation}
with an arbitrary error probability $\prob_e \leq \epsilon''$, such that both $\eta,\epsilon''$ diminish at large~$n$.
 \label{lemma:ptp_capacity}
 \end{lemma}

\begin{proof}		
		
{\em Encoding:} Nested lattice codes are used, where the coarse lattice $\lat \in \real^{M n}$ has second moment~$\Px'$. The codeword is composed of $n$~super-symbols $\xv_i$ each of length~$M$, as shown in \eqref{codeword_MIMO}, which are transmitted throughout $n$ channel uses.

{\em Decoding:} The received signal can be expressed in the form $\yv=\Hm_s \xv+\wv$, where  $\Hm_s$ is a block-diagonal matrix whose diagonal block~$i$ is $\Hm_i$. The received signal $\yv$ is multiplied by $\Um_s \in \real^{N n \times M n}$ and the dither is removed as follows
\begin{align}
\yv' \triangleq & \Um_s^T \yv + \dv \nonumber \\
 = & \tv + \latpoint + \zv,
\label{sig_rx_2_MIMO}
\end{align}
where 
\begin{equation}
\zv \triangleq ( \Um_s^T \Hm_s - \Ix_{M n}) \xv + \Um \wv,
\label{eq_noise_MIMO}
\end{equation}
and $\tv$ is independent of $\zv$, according to Lemma~\ref{lemma:uniform}. $\Um_s$ is then a block-diagonal matrix, where the $M \times N$ equalization matrix at time~$i$ is the MMSE matrix given by
\begin{equation}
\Um_i= \Px'  (\Ix_{N} + \Px' \Hm_i \Hm_i^T)^{-1} \Hm_i.
\label{eq:U_MSE_MIMO}
\end{equation}

From \eqref{eq_noise},\eqref{eq:U_MSE_MIMO}, $\zv_i \in \real^{M}$ is expressed as
\begin{equation}
\zv_i = - ( \Ix_{M} + \Px' \Hm_i^T \Hm_i  )^{-1} \xv_i + \Px'  \Hm_i^T ( \Ix_{N} + \Px' \Hm_i \Hm_i^T )^{-1} \wv_i,
\label{eq:zi_MIMO}
\end{equation}
where  $\zv \triangleq [\zv_1^T, \ldots , \zv_n^T]^T$.
We apply a version of the ambiguity decoder proposed in~\cite{Loeliger}, defined by an ellipsoidal decision region $\genvoronoi \in \real^{M n}$, as follows
\begin{equation}
 \genvoronoi \triangleq \Big \{ \vv \in \real^{M n}~:~ \vv^T  \boldsymbol{\cov}_s^{-1} \vv \leq (1+\gamma) M n  \Big \},
 \label{voronoi_ptp_MIMO_opt}
 \end{equation}
 where $\boldsymbol{\cov}_s$ is a block-diagonal matrix, whose diagonal block~$i$, $\boldsymbol{\cov}_i$, is given by
 \begin{equation}
 \boldsymbol{\cov}_i \triangleq \Px' \big( \Ix_{M} + \Px' \Hm_i^T \Hm_i  \big)^{-1}.
 \label{cov_MIMO_opt}
 \end{equation}
 Let $\SNRm_i \triangleq \Ix_{M} + \Px' \Hm_i^T \Hm_i$. The volume of~$\genvoronoi$ is then
 \begin{equation}
 \vol(\genvoronoi) = (1+ \gamma)^{\frac{M n}{2}}  \vol \big( \mathcal{B}_{M n} (\sqrt{M n \Px'}) \big)  \,  \prod_{i=1}^n \det( \SNRm_i)^{\frac{-1}{2}}. 
 \label{vol_omega_MIMO_opt}
 \end{equation}

{\em Error Probability:} 
As shown in~\cite[Theorem~4]{Loeliger}, on averaging over the ensemble of fine lattices~$\Cs$ of rate~$R$ that belong to the class proposed in Section~\ref{sec:lattice}, 
\begin{align}
\frac{1}{|\Cs|} \sum_{\Cs} \, \prob_e  & <    \,  \prob_{\text{out}}^{(\eta)}  + \prob (\zv \notin \genvoronoi) +  (1+ \delta) \, \frac{\vol(\genvoronoi)}{\vol(\voronoi_1)} 
\nonumber \\
& =  \, \prob_{\text{out}}^{(\eta)} + \prob (\zv \notin \genvoronoi) + (1+ \delta) 2^{nR} \, \frac{\vol(\genvoronoi)}{\vol(\voronoi)}, 
\label{error_prob_MIMO}
\end{align}
for any $\delta > 0$, where~\eqref{error_prob_MIMO} follows from~\eqref{lattice_rate}. This is a union bound involving
{\em three} events: the event that the average throughput achieved by the channel sequence is bounded away from its statistical mean by more than~$\eta$, and the event that the noise vector is outside the decision region, i.e., $\zv \notin \genvoronoi$ and the event that the post-equalized point is in the intersection of two decision regions, i.e., $\big\{\yv' \in \{ \tv_1 + \genvoronoi \} \cap \{ \tv_2 + \genvoronoi\} \big\}$, where $\tv_1,~\tv_2 \in \lat_1$ are two distinct lattice points. 
 From~\eqref{eq:rateMean}, $\prob_{\text{out}}^{(\eta)}<\tilde{\epsilon}$ for any $\tilde{\epsilon}>0$ at large~$n$. Following in the footsteps in Appendix~\ref{appendix:z_dist}, $\prob (\zv \notin \genvoronoi)< \hat{\epsilon}$ for any $\hat{\epsilon} > 0$ for large~$n$. Let $\epsilon' \triangleq \tilde{\epsilon} + \hat{\epsilon}$. The error probability can then be bounded by 
\begin{equation}
\epsilon'' \triangleq \frac{1}{|\Cs|} \sum_{\Cs} \, \prob_e < \epsilon' + (1+ \delta) 2^{nR} \frac{\vol(\genvoronoi)}{\vol(\voronoi)}, 
\label{error_prob_2_MIMO}
\end{equation}
for any $\gamma,\delta>0$. 
The second term in~\eqref{error_prob_2_MIMO} is then 
\begin{equation}
\epsilon_{\text{avg}} \triangleq 2^{ -  n \Big (  -  R +  \frac{1}{2 n} \sum_{i=1}^{n}  \log \det (  \SNRm_i ) - \epsilon''' \Big ) } ,
\label{eq:exponential_MIMO}
\end{equation}
where 
\begin{equation}
\epsilon'''  \triangleq   \frac{1}{ n} \log \big( \frac{\vol(\ball_{M n} (\sqrt{M n \Px'}))}{\vol(\voronoi)} \big) + \log({1+ \gamma})^{\frac{M}{2}} + \frac{1}{n}  \log  (  1+\delta )
\label{xi_MIMO}
\end{equation}
From~\eqref{covering}, the first term in~\eqref{xi_MIMO} vanishes, and so do the second and third terms as $n$ increases.  The probability of error averaged over the codebooks in $\Cs$ is bounded by~
\begin{equation}
\epsilon'' \triangleq \epsilon' + \epsilon_{\text{avg}}.
\label{error_averaging}
\end{equation}
Then there exists at least one codebook that achieves $R <  \frac{1}{2 n} \sum_{i=1}^{n}  \log \det (  \SNRm_i )$, which converges to~\eqref{capacity_eqn}. The remainder of the proof follows Section~\ref{sec:heuristic}. 
\end{proof}

Note that Lemma~\ref{lemma:ptp_capacity} does not imply the rate in~\eqref{capacity_eqn} is {\em universally} achievable, since it does not guarantee the existence of a single codebook that achieves this rate for all fading sequences drawn from an underlying distribution. A similar approach was adopted in~\cite{paper2}, whose universality is not conclusive. In the sequel we discuss the rates achievable using universal codebooks. Similar to Section~\ref{sec:CSIT} we first address channels with finite-support fading distributions and then extend the result to continuous-valued, unbounded fading.


\subsection{Finite-Support Fading Distributions}
\label{sec:CSIR_discrete}

We address point-to-point block-fading channels with coherence length~$b$, whose channel coefficients are drawn from a discrete distribution with finite-support~$\ss$. The following result utilizes the weak typicality arguments provided in Section~\ref{sec:typicality} to show the existence of a nested pair of lattice codes that achieve rates within a constant gap to ergodic capacity.

 \begin{theorem}
 For a stationary and ergodic block-fading $M \times N$ MIMO channel with coherence interval $b$ and fading coefficients drawn from a finite-support distribution~$\prob_h$, a universal nested lattice code exists that achieves rates within a constant gap $\gap \triangleq \frac{M N}{b} \hbar(H)$ bits per channel use of the ergodic capacity, where~$\hbar(H)$ is the entropy rate of the fading process.
 \label{theorem:cap_discrete}
 \end{theorem}

\begin{proof}

 Lemma~\ref{lemma:ptp_capacity} ensures the existence of one codebook in~$\Cs$ that achieves the rate in~\eqref{capacity_eqn} with error probability that is less than~$\epsilon''$. We now show that if we allow a  multiplicative increase in the error probability, numerous codebooks in~$\Cs$ can support the rate~$R$ in~\eqref{capacity_eqn}.
 \begin{lemma}
 For the channel under study in Lemma~\ref{lemma:ptp_capacity}, at least~$\frac{\kappa-1}{\kappa} \, |\Cs|$ codebooks in~$\Cs$ achieve the rate~$R$ in~\eqref{capacity_eqn} with at most~$\kappa \epsilon''$ error probability, for any $\kappa \in \ints^+$ where~$\kappa<|\Cs|$.
 \label{lemma:count}
 \end{lemma}

 \begin{proof} 
 We expurgate codebooks from~$\Cs$ as follows. First, arrange the codebooks in  descending order of the error probability $\epsilon_1, \ldots, \epsilon_{|\Cs|}$on  the MIMO channel defined in Lemma~\ref{lemma:ptp_capacity}. Then, discard the first~$\frac{1}{\kappa} |\Cs|$ codebooks. From~\eqref{error_averaging}, the error probability of each of the remaining $\frac{\kappa-1}{\kappa} |\Cs|$ codebooks is then bounded by~$\epsilon' + \kappa \epsilon_{\text{avg}}$, as follows\,%
 \footnote{$\epsilon'$ in~\eqref{error_prob_2_MIMO} is independent of the codebook, so the average over codebooks is also $\epsilon'$.}
 \begin{align}
 |\Cs| \epsilon_{\text{avg}} = & \,  \sum_{\ell=1}^{|\Cs|/\kappa} \epsilon_\ell \, + \epsilon_{1+|\Cs|/\kappa} \, + \, \sum_{\ell=2+|\Cs|/\kappa}^{|\Cs|} \epsilon_\ell  \, 
 \nonumber \\
 \geq & \, \sum_{\ell=1}^{|\Cs|/\kappa} \epsilon_\ell  + \, \epsilon_{1+|\Cs|/\kappa}   \,
 \geq \, (1+\frac{1}{\kappa}|\Cs|) \, \epsilon_{1+|\Cs|/\kappa} \, ,
 \label{bounding}
 \end{align}    
 Hence, 
 \begin{equation}
 \epsilon_{1+|\Cs|/\kappa} \, \leq \frac{\kappa}{1+\frac{\kappa}{|\Cs|}} \, \epsilon_{\text{avg}} \, < \, \kappa \, \epsilon_{\text{avg}}.
 \label{bounding_final}
 \end{equation} 
 Since~$\epsilon_{l+|\Cs|/\kappa} \leq \epsilon_{1+|\Cs|/\kappa}$ for any~$\ell>1$, each of the last~$\frac{\kappa-1}{\kappa} \, |\Cs|$ codebooks in~$\Cs$ have error probability that does not exceed~$\epsilon' + \kappa \epsilon < \kappa \epsilon''$. 
 \end{proof}

 To summarize, Lemma~\ref{lemma:count} shows that given a channel matrix~$\Hm_s$, a constant fraction of all codebooks in~$\Cs$ achieves the rate in~\eqref{capacity_eqn}, e.g., for~$\kappa=100$, at least $99\%$ of the codebooks in~$\Cs$ incur no more than~$100 \epsilon''$ error probability, where~$\epsilon''$ can be made arbitrarily small by increasing~$n$. Note that the proof technique in Lemma~\ref{lemma:count} is not limited to lattice codes, and can be used for other random ensembles of codebooks. 

Now, assume a stationary and ergodic block-fading MIMO channel with~$n'$ blocks, whose $MNn'$ distinct channel coefficients are drawn according to a distribution~$\prob(H)$ with a finite support of size~$|\ss|$. The $\epsilon$-weakly typical set~$T_{\epsilon}^{(W)} (H)$ of channel sequences~$\Hm_s$ is denoted hereafter by $T_{\epsilon}^{(H)}$. We aim at answering the following question: under what rates can a single codebook in~$\Cs$ achieve vanishing error probability for all channel sequences~$\Hm_s \in T_{\epsilon}^{(H)} \, ?$ 

 Denote by~$\Cs_j$ the set of codebooks that achieve at most~$\kappa \epsilon''$ error probability for the channel matrix~$\Hm_s^{(j)}$ indexed by~$j$. Recall the cardinality of each of these sets is~$\frac{\kappa-1}{\kappa} \, |\Cs|$. The event that no codebook is universal over typical channel sequences can be represented by either of the following two conditions:
 \begin{align}
 \Cs_1 \cap \Cs_2 \cap \ldots \cap \Cs_{|T_{\epsilon}^{(H)}|} = \, \phi,     \nonumber \\
 \Cs_1^c \cup \Cs_2^c \cup \ldots \cup \Cs_{|T_{\epsilon}^{(H)}|}^c = \, \Cs,  
 \label{channel_intersection}
 \end{align}
 where $\mathcal{A}^c$ denotes the complement of the set~$\mathcal{A}$. Hence, it can be shown via a union bound that 
 \begin{equation}
 |\Cs_1^c \cup  \ldots \cup \Cs_{|T_{\epsilon}^{(H)}|}^c| \leq \frac{1}{\kappa} \, |\Cs| \, |T_{\epsilon}^{(H)}|,
 \label{union_bounding}
 \end{equation}
  Hence, from~\eqref{channel_intersection} a universal codebook with negligible error probability is guaranteed to exist, as long as~$\kappa> |T_{\epsilon}^{(H)}|$. On substituting in~\eqref{error_prob_2_MIMO},\,\eqref{eq:exponential_MIMO},
 \begin{align}
\prob_e = &  \,   \prob \big(e| \Hm_s \in T_{\epsilon}^{(H)} \big) \, \prob \big( \Hm_s \in T_{\epsilon}^{(H)} \big) 
\twocolbreak
\, + \,  \prob \big(e | \Hm_s \notin T_{\epsilon}^{(H)} \big) \, \prob \big( \Hm_s \notin T_{\epsilon}^{(H)} \big) \nonumber \\
\leq &  \,   \prob \big(e| \Hm_s \in T_{\epsilon}^{(H)} \big)  \, + \, \prob \big( \Hm_s \notin T_{\epsilon}^{(H)} \big) \nonumber \\    
  < &   \, \kappa \epsilon \, + \, \epsilon'  \, + \, \epsilon_T   \nonumber \\
< & \,  2^{-n \big( -R+\frac{1}{2} \Ex \big[ \log \det (\Ix_{M}+ \Px' \Hm^T \Hm) \big]  - \frac{1}{n} \log |T_{\epsilon}^{(H)}|  -\epsilon''' - \eta \big) }
\twocolbreak
\, + \, \epsilon' \, + \epsilon_T,
\label{prob_error_universal}
\end{align}
where $\prob \big( \Hm_s \notin T_{\epsilon}^{(H)} \big)\triangleq \epsilon_T$. Based on weak typicality arguments in Section~\ref{sec:typicality}, $|T_{\epsilon}^{(H)}| \leq 2^{MNn'\big( \hbar(H) + \epsilon_T \big)}$. Hence, reliable rates can be achieved as long as
 \begin{equation}
 R < \frac{1}{2} \Ex \big[ \log \det (\Ix_{M}+ \Px' \Hm^T \Hm) \big] - \frac{MN}{b} \big( \hbar(H) + \epsilon_T \big) -\epsilon''' - \eta \, ,  
 \label{rate_universal}
 \end{equation} 
 where $\epsilon''',\epsilon_T$ and $\eta$ can be made arbitrarily small by increasing~$n$. Since the third term in~\eqref{prob_error_universal} diminishes with~$n$, there exists $n_{\epsilon} \in \ints^+$ such that for all $n> n_{\epsilon}$, $\prob_e$ in~\eqref{prob_error_universal} satisfies $\prob_e < 2 \epsilon'$. The final step to complete the proof is showing that the number of possible channel matrices~$\Hm_s$ does not exhaust~$|\Cs|$, otherwise~$\kappa> |T_{\epsilon}^{(H)}|$ cannot be guaranteed. From the lattice construction in~\cite[Section III]{Lattices_good},  there exists at least~$q^n$  generator matrices that generate unique lattices, where~$q$ is the size of the prime field from which the lattice is drawn. Since $\sqrt{n}/q \to 0$ as $n \to \infty$, a lower bound on~$|\Cs|$ is~$n^{n/2}$. Since the number of possible channels cannot exceed~$|T_{\epsilon}^{(H)}| < 2^{MNn \big( \hbar(H) + \epsilon_T \big)}$ where~$2^{MN \big( \hbar(H) + \epsilon_T \big)}$ is non-increasing with~$n$, there exists~$\check{n} \in \ints^+$ such that for all $n>\check{n}$, $n^{n/2} > 2^{MNn \big( \hbar(H) + \epsilon_T \big)}$. Hence, $\kappa> |T_{\epsilon}^{(H)}|$ is guaranteed for large enough~$n$.

The previous result depicts the gap to capacity for channels whose optimal input signal covariance is $ \Px' \Ix_M$. The extension to channel distributions whose optimal input covariance is non-white, i.e., not a scaled identity, is straightforward. Let~$\covx^*$ denote the optimal input covariance matrix, i.e., $\covx^* = \text{arg} \max ~ C$ given in~\eqref{capacity_MIMO_gen}. The transmitted codeword is then  
\begin{equation}
\check{\xv} \triangleq [\covx^{* \frac{1}{2}} \xv_1^T  ,  \ldots ,  \covx^{* \frac{1}{2}} \xv_n^T]^T,
\label{codeword_MIMO_gen}
\end{equation}
where~$\xv$ is drawn from a nested lattice code whose coarse lattice~$\lat \in \real^{M n}$ has a unit second moment. Hence, the received signal can be expressed by  
\begin{align}
\yv_i & =  \Hm_i \covx^{* \frac{1}{2}} \xv_i + \wv_i   \nonumber \\
      & \triangleq  \check{\Hm}_i   \xv_i + \wv_i,
\label{sig_Rx_MIMO_gen}
\end{align}
where~$\check{\Hm}_i \triangleq \Hm_i \covx^{* \frac{1}{2}}$. From~\eqref{sig_Rx_MIMO_gen} the following rates are achievable
\begin{align*}
R <& \,  \frac{1}{2} \Ex \Big[ \log{ \det \big ( \Ix_{M} +  \check{\Hm}^T \check{\Hm} \big ) } \Big] - \frac{MN}{b} \hbar(H) \\ 
  =& \,  \frac{1}{2} \Ex \Big[ \log{ \det \big ( \Ix_{N} + \Hm \covx^* \Hm^T \big ) }  \Big] - \frac{MN}{b} \hbar(H) ,
\end{align*}
which is the optimal value of the expression in~\eqref{capacity_MIMO_gen}. This concludes the proof of Theorem~\ref{theorem:cap_discrete}.	
\end{proof}

\begin{corollary}
The gap to capacity in Theorem~\ref{theorem:cap_discrete} can be bounded from above by $ \frac{M N}{b} \log |\ss|$ bits per channel use, where $\ss$ is the supporting set of the channel coefficients.
\label{corollary:cardinality}
\end{corollary}

\begin{proof}
The proof follows directly from Theorem~\ref{theorem:cap_discrete} since $\hbar(H) \leq \log |\ss|$.
\end{proof}

\begin{figure}
\centering
\includegraphics[width=\Figwidth]{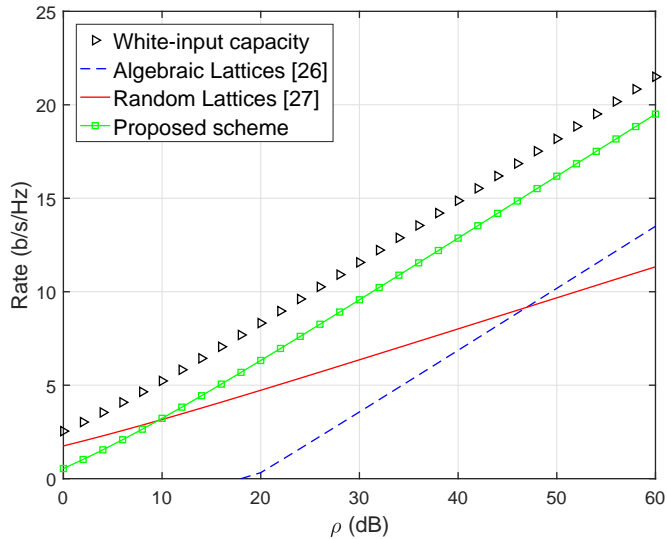}
\caption {Achievable rates for the $2 \times 2$ discrete-fading MIMO channel.} 
\label{fig:MIMO} 
\end{figure}

The performance of the proposed scheme is investigated under a $2 \times 2$~MIMO block-fading channel with~$b=20$ and~$|\ss| = 10^3$, where the channel coefficients are drawn uniformly and independently from uniformly-spaced values within range $[-5,5]$. For this scenario $\gap <2$ bits per channel use. Rates are plotted in Fig,~\ref{fig:MIMO} and compared with the white-input capacity, as well as the rates achieved by the algebraic lattice scheme in~\cite{Luzzi_j} and the lattice coding scheme in~\cite{Journal1}. The proposed scheme is shown to outperform the baseline schemes for moderate and high SNR values.


\subsection{Continuous-Valued, Unbounded Fading Distributions}
\label{sec:CSIR_cont}

\begin{figure}
\centering
\includegraphics[width=\Figwidth]{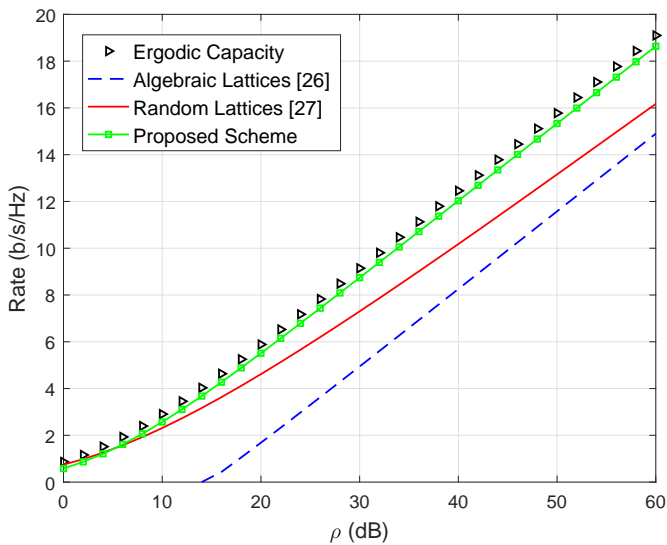}
\caption {Comparison of the achievable rates under Rayleigh fading.} 
\label{fig:Rayleigh} 
\end{figure}

The techniques developed in Section~\ref{sec:CSIR_discrete} produce rates within a gap to capacity that is tied to the number of distinct fading states, and hence cannot be directly applied to continuous-valued channels. The main idea of this section is a refinement of the technique developed in Section~\ref{sec:continuous1}, where the scheme is designed for a quantized version of a continuous-valued channel so that it can also work on the continuous-valued channel itself.
The quantized channel gains are derived by mapping the continuous channel gains in each quantization interval to the lower end of that interval.
In this manner the magnitude of each continuous-valued channel coefficient is larger than its quantized counterpart, so it is easy to show that the encoder/decoder performance over the continuous channel is no worse than the performance over the quantized channel.
The disadvantage of this strategy is that it forfeits the power in the tail of the fading distribution, and also within each quantization interval it cannot utilize the fading gain that is beyond the lower end of that interval. Unlike the perfect CSIT case in Section~\ref{sec:continuous1}, the quantization cannot be made arbitrarily fine since the universality penalty that follows from the discrete channel analysis is proportional to the number of fading states. We now develop the strategy in more detail and show that with careful optimization of the quantization intervals, one can produce performance guarantees that come very close to capacity.

Concretely, the scheme is similar to the discrete-valued fading case
up to the MMSE scaling in~\eqref{eq:U_MSE_MIMO}, where the MMSE scaling is pursued using the continuous-valued channel coefficients. The channel magnitude~$|h_i|$ is then quantized into~$L+1$ sets $Q_\ell \triangleq [q_{\ell-1},q_{\ell}]$, where  $q_0\triangleq 0$ and $q_{L+1} \triangleq \infty$, and each channel coefficient is quantized to the lower limit of the bracket $Q_\ell$ in which it belongs, producing a discrete valued sequence. 
The gap to capacity $\gap \triangleq C-R$ is then bounded as follows\,%
\begin{align}
\gap \leq & \, \Ex \big[ \log (1+ \Px |h|^2) \big] \, - \,  \Ex \big[\log (1+ \Px q^2) \big] + \, \frac{1}{b} \log \big( L+1  \big)   \nonumber \\
= & \frac{1}{b} \log \big( L+1  \big)  + \Ex \big[\log \Big(\frac{1+ \Px |h|^2}{1+\Px q_L^2} \Big) \big| |h| > q_L \big] \prob \big(|h| > q_L \big) 
\nonumber \\
& + \sum_{\ell=1}^L  \Ex \big[\log \Big(\frac{1+ \Px |h|^2}{1+\Px q_{\ell-1}^2}\Big) \big| |h| \in Q_\ell \big] \prob \big(|h| \in Q_\ell \big),
\label{eqn_gap_continuous_CSIR}
\end{align}
where~$q$ is a random variable with support~$\{q_0, \ldots,q_L\}$ representing the quantized channel magnitude. The first term in~\eqref{eqn_gap_continuous_CSIR} follows from Corollary~\ref{corollary:cardinality}.
The minimization of the gap to capacity requires the design of a quantizer with $L+1$~levels as mentioned above that minimizes the sum of the three terms in~\eqref{eqn_gap_continuous_CSIR}. 
Recall in Section~\ref{sec:continuous1} we used a uniform quantizer, which is not necessarily optimal but was sufficient for our purposes at that point. The optimal quantizer for~\eqref{eqn_gap_continuous_CSIR} may be derived by calculating the slope of the cost function with respect to each $q_\ell$ and forcing these slopes to be equal (KKT conditions). In iterative optimization one would update in each iteration the $q_\ell$ with the largest slope. Since the cost function involves integration over fading, the cost slopes can be calculated using the Leibnitz integration rule. These calculations are somewhat cumbersome, motivating us to find a more easily described quantizer that also experimentally yields a small gap after optimization. To this end, we design a quantizer so that the probability of the fading coefficient falling into each of its bin is equal, i.e., $\prob \big( |h| \in Q_1  \big)  = \ldots =   \prob \big( |h| \in Q_L  \big)$. This results in the following simplification in the cost function
\begin{align}
\gap \leq & \frac{1}{b} \log \big( L+1  \big) + \gamma_{L} \, \Ex \big[\log \Big(\frac{1+ \Px |h|^2}{1+\Px q_L^2} \Big) \big| |h| > q_L \big]  \nonumber \\
& \, +  \frac{1- \gamma_{L}}{L}  \, \sum_{\ell=1}^L  \Ex \big[\log \Big(\frac{1+ \Px |h|^2}{1+\Px q_{\ell-1}^2}\Big) \big| |h| \in Q_\ell \big] ,
\label{eqn_gap_continuous_CSIR_equal}
\end{align}
where $\prob \big(|h| > q_L \big) \triangleq \gamma_{L}$ is the  complement of the cumulative distribution function of~$|h|$. Note that $\gamma_{L} = e^{-q_L^2}$ for Rayleigh fading with normalized gain.
Once this structure is fixed, only two parameters~$L, \, q_L$ need to be optimized. 
In our simulations, we use two-dimensional grid search to find~$L, \, q_L$ that minimize the gap expression in~\eqref{eqn_gap_continuous_CSIR_equal}. The rates achieved under Rayleigh block fading with~$b=20$ are plotted in Fig.~\ref{fig:Rayleigh}, which demonstrates that the gap to capacity is within 0.5~bits per channel use up to SNR of $60$~dB.%
\footnote{Although Fig.~\ref{fig:Rayleigh} indicates that the proposed scheme achieves close-to-capacity rates under Rayleigh fading, it remains unverified whether it is within a constant gap to capacity for all SNR.}


\section{Conclusion} 
\label{sec:conc}

This paper demonstrates that a precoded lattice coding scheme achieves the capacity of the fading point-to-point channel with channel state information at both the transmitter and the receiver. A key difference with earlier ergodic fading results is non-separable coding. Furthermore, the decision regions are fixed for a given channel distribution. The proposed scheme is first discussed in the context of a heuristic channel model. The results are then extended to arbitrary MIMO channel distributions with robustly typical realizations, and to continuous-valued channels. With CSIR but no CSIT, an alternative decoding strategy is presented for block-fading MIMO channels, where channel-matching decision regions are proposed. Under fading drawn from a discrete distribution with finite support, the achieved rates are within an SNR-independent gap to capacity. The scheme is also extended to continuous-valued fading where it is shown that achievable rates approach capacity under Rayleigh fading.


\section*{Acknowledgment}
The authors acknowledge the associate editor and the reviewers for useful comments that led to tightening the results in Section~\ref{sec:CSIR}.

\appendices

\section{Precoded Signal Satisfies the Power Constraint}
\label{appendix:power}

We first present the following lemma, whose proof can be found in~\cite{weighted_avg}.

\begin{lemma} \cite[Theorem~1]{weighted_avg}.
Let~$\{ \beta_1, \ldots, \beta_n \}$ be a monotonically increasing sequence of finite positive weights, where $\displaystyle{\lim_{n \to \infty}} \sum_{i=1}^n \beta_i = \infty$. Then, for a sequence of random variables $q_1, \ldots, q_n$ with mean~$\mu_q$, 
\begin{equation}
 \Ex \Big[ \frac{1}{\sum_{\ell = 1}^n \beta_{\ell}} \sum_{i=1}^n \beta_i q_i   \Big] < \mu_q + \epsilon \, ,
\label{weightedmean}
\end{equation}   
where $\epsilon>0$ vanishes with~$n$.
\label{lemma:weighted_avg}
\end{lemma}

We now show that for $\xv'  = \Dm \Vm \xv$,  $\frac{1}{n} \Ex \big[ ||\xv'||^2 \big] \, \approx \, \frac{1}{n} \Ex \big[ ||\xv||^2 \big]$ at large~$n$.

\begin{align}
\frac{1}{n} \, \Ex [||\xv'||^2] \, & = \, \frac{1}{n} \, \Ex \big[ \xv^T \Vm^T \Dm^2 \Vm \xv \big] \nonumber \\
& = \, \frac{1}{n} \, \Ex_{\Wm} \Big [ \Ex_{\xv|\Wm} \big [\xv^T \Wm \xv  \big] \Big] \nonumber  \\
& = \, \Ex_{\Wm} \Big [ \Ex_{\xv|\Wm} \big [ \frac{1}{n} \, \sum_{i=1}^n w_i x_i^2  \big] \Big] \,  \label{power1}   \\
& < \, \Ex_{\Wm} \Big [ \Ex_{\xv|\Wm} \big [ \frac{1}{n} \, \sum_{i=1}^n \tilde{w}_i x_i^2 \big ] \Big] \,  \label{power2}   \\
& = \, \Ex_{\Wm} \Big [ \frac{1}{n}  \displaystyle{\sum_{\ell=1}^n} \tilde{w}_{\ell} ~ \Ex_{\xv|\Wm} \big [ \frac{1}{\sum_{\ell=1}^n \tilde{w}_{\ell}}  \,  \sum_{i=1}^n \tilde{w}_i x_i^2 \big ] \Big]  \nonumber  \\
& \leq \, ( \Px + \epsilon_1 ) \, \Ex_{\Wm} \Big [  \frac{1}{n} \sum_{i=1}^n \tilde{w}_i  \Big] \,  \label{power3}   \\
& = \, ( \Px + \epsilon_1 ) \, (  1 + \epsilon_2 +  \epsilon_3 ) \, \triangleq  \,  \Px  + \epsilon' \, . \label{power4}   
\end{align}
Given the structure of the permutation matrix~$\Vm$ described in Appendix~\ref{appendix:V}, $\Wm \triangleq \Vm^T \Dm^2 \Vm$ is a diagonal matrix with non-decreasing positive entries $w_1, \ldots, w_n$, and hence~\eqref{power1} follows. Note that~$\xv$ and~$\Wm$ are independent. Define $\tilde{w}_i \triangleq w_i + \delta_i$, where $\delta_i \triangleq  \frac{i}{n} e^{-n} \upsilon$, and $\upsilon \triangleq \displaystyle{\min_{ \{j,k\} \, , \, w_j \neq w_k }}  \Big | w_j - w_k \Big|$. Hence, from~\eqref{power2}, $\delta_i$ assures all weights $\tilde{w}_i$ are positive and monotonically increasing. \eqref{power3} follows from Lemma~\ref{lemma:weighted_avg}. \eqref{power4} follows since $\frac{1}{n} \sum_{i=1}^n  w_i < 1 + \epsilon_2$, where $w_i$ represents the {\em normalized} waterfilling power allocations, and  $\epsilon_3 \triangleq \, \frac{1}{n} \sum_{i=1}^n  \delta_i$ vanishes with~$n$. Thereby $\epsilon' \to 0$ as~$n \to \infty$.


\section{Designing the Permutation Matrix~$\Vm$}
\label{appendix:V}

As mentioned in Section~\ref{sec:heuristic}, the role of~$\Vm$ is re-ordering the channel coefficients in the diagonal matrix~$\Hm$ in ascending order of the magnitudes. Recall the coefficients take on the values~$\hh_1,\ldots,\hh_{|\ss|}$, where~$\hh_k$ appears with frequency~$\mu_k$. 
Now, define a counter~$\nu_k^{(i)} \leq n \mu_k$, that counts the number of occurrences of~$\hh_k$ after channel use~$i$. Typically, when $h_i=\hh_k $, then $\Vm(i,:)$ is $\impulse_m^T$, where $m=\nu_k^{(i)}+ \sum_{l=1}^{k-1} n \mu_l$. It easily follows that $\Vm^T \Hm \Vm = \Hm_{\order}$ for any~$\Hm$ that belongs to the random location channel model.


\section{Proof of Lemma~\ref{lemma:z_dist}}
\label{appendix:z_dist}

We follow in the footsteps of \cite{Erez_Zamir,DMT_lattices}. Consider a noise vector~$\zv^* \in \real^n$ that is closely related to the post-equalizer noise $\zv$ as follows%
\begin{equation}
\zv^* =  \, \Am \gv + \Bm \big( \sqrt{\Px} \, \wv +  \sqrt{\sigma_{\ball}^2 - \Px} \, \wv^* \big) \, ,
\label{eq:zi_2}
\end{equation}
where~$\gv, \wv^*$ are i.i.d. Gaussian with zero mean and covariances~$\sigma_{\ball}^2 \Ix_n$, $\Ix_n$, respectively and $\sigma_{\ball}^2$ is the second moment of the smallest sphere covering~$\voronoi$. $\Am, \Bm$ are diagonal matrices whose diagonal elements are as follows
\begin{equation}
A_{ii} \triangleq \frac{-1}{\Px^*_{\order(i)}  h_{\order(i)}^2 +1} ~~ , ~~ 
B_{ii} \triangleq \frac{\sqrt{ \Px^*_{\order(i)} } h_{\order(i)}}{\Px^*_{\order(i)} h_{\order(i)}^2 +1} \, .
\label{AABB}
\end{equation} 
 It is then easy to show that the auto-correlation matrix~$\boldsymbol{\cov}^*$ of $\zv^*$ is diagonal, whose elements are given by
\begin{equation}
\boldsymbol{\cov}^*_{ii} = \frac{\sigma_{\ball}^2}{\Px^*_{\order(i)}  h_{\order(i)}^2 +1}  \, .
\label{cov_noise2}
\end{equation}
Note that $\sigma_{\ball}^2 = (1+\epsilon') \Px$, where~$\epsilon'$ can be made arbitrarily small by increasing~$n$~\cite[Lemma 6]{Erez_Zamir}. 
 The probability $\prob \big( \zv^* \notin \genvoronoi \big)$ is then equivalent to $\prob \big( ||\zv^{(w)}||^2 > (1+\epsilon'') n \big)$, where $\zv^{(w)} \triangleq \boldsymbol{\cov^{* \frac{-1}{2}}} \zv^*$. Hence, $||\zv^{(w)}||^2$ is a  chi-squared random variable with~$n$ degrees-of-freedom. Using the Chernoff bound~\cite{kobayashi},
\begin{align}
\prob \big( ||\zv^{(w)}||^2 > (1+\epsilon'') n \big) \leq 
& \min_{t \geq 0}  \big \{  e^{- n \big( (1+\epsilon'')t + \log_e (1-t) \big)} \big \}  \nonumber \\
= & (1+\epsilon'')^{\frac{ n}{2}} e^{-\frac{ n \epsilon''}{2}}    \nonumber \\ 
= & e^{-\frac{ n}{2} \big( \epsilon'' - \log_e (1+ \epsilon'')  \big)}.
\label{chernoff}
\end{align}

Now, we show that the probability density of $\zv$, $f_{\zv} (\vv)$ is upper-bounded (up to a constant) by $f_{\zv^*}(\vv)$. It was shown in~\cite[Lemma 11]{Erez_Zamir}, that 
\begin{equation}
f_{\xv} (\vv) \leq e^{c_n} f_{\gv} (\vv),
\label{pdf_bound}
\end{equation}
where $c_n/n \to 0$ as $n \to \infty$. Hence,  $f_{\Am \xv} (\vv) \leq e^{\delta' n} f_{\Am\gv} (\vv)$ and $f_{\zv} (\vv) \leq e^{\delta' n} f_{\zv^*} (\vv)$ follow as well, where the former inequality is obtained using transformation of random variables whereas the latter inequality is obtained via convolution of both terms in~\eqref{eq:zi_2}. Hence,
\begin{align}
 \prob (\zv \notin \genvoronoi) & \, =  \,  \int_{v \notin \genvoronoi} {f_{\zv}(\vv) dv} \nonumber \\ 
& \, \leq \,  e^{c_n} \,  \int_{v \notin \genvoronoi} {f_{\zv}(\vv) dv}  \nonumber \\
& \, = \, e^{c_n} \, \prob \big( ||\zv^{(w)}||^2 > (1+\epsilon'') n \big) \nonumber \\ 
& \, \leq \, e^{-\frac{n}{2} \big( \epsilon'' - \log_e (1+ \epsilon'') - \frac{2 c_n}{n}  \big)},
\label{error_bound}
\end{align}
where $\log_e$ is the natural logarithm. Since $\epsilon'' > \log_e (1+ \epsilon'')$ for all $\epsilon''>0$ and $c_n/n \to 0$ as $n \to \infty$, the exponent in~\eqref{error_bound} remains negative and it can be shown that there exists $n_{\gamma}$ such that for all $n> n_{\gamma}$, $ \prob (\zv \notin \genvoronoi)< \gamma$. 
The final step is to show that the elements of $\sqrt{\sigma_{\ball}^2 - \Px} \, \Bm \wv^*$ vanish with~$n$. Let $\bar{\wv} \triangleq \Bm \wv^*$, and $\gamma^* \triangleq \sqrt{\sigma_{\ball}^2 - \Px}$. Since~$|B_{ii}|<1$ and $\wv^*$ is i.i.d. with unit variance, the variance of each of the elements~$\bar{w}_i$ is no more than~$1$. Using Chebyshev's inequality~\cite{kobayashi}, 
\begin{align}
\prob \big(\gamma^*\sqrt{\sigma_{\ball}^2 - \Px} \, \bar{w}_i \geq \gamma^* \, \kappa \big) & \leq \frac{1}{\kappa^2}, ~~~~~~ for~all~ \kappa > 0,
\nonumber \\
\prob \big(\sqrt{\sigma_{\ball}^2 - \Px} \, \bar{w}_i \geq \sqrt{\gamma^*} \,  \big) & \leq \gamma^* \,, 
\label{chebyshev}
\end{align}
and~\eqref{chebyshev} follows when $\kappa =\frac{1}{\sqrt{\gamma^*}}$. Since $\displaystyle{\lim_{n \to \infty}} \sigma_{\ball}^2 = \Px$ for a covering-good lattice, the elements of $\sqrt{\sigma_{\ball}^2 - \Px} \, \Bm \wv^*$ vanish with~$n$. This concludes the proof of Lemma~\ref{lemma:z_dist}.


\section{Proof of Theorem~\ref{theorem:rate_ptp}}
\label{appendix:theorem_ptp}

We bound the total number of channel occurrences that deviate from $n \prob_k$ 
in~\eqref{robust_1} as follows\,%
\footnote{For simplicity, we assume $n \prob_k$ and $n \delta$ are positive integers.}
\begin{equation}
\sum_{k=1}^{|\ss|} |n_k - n \prob_k| \leq  \sum_{k=1}^{|\ss|} \delta n \prob_k = \delta n \triangleq n_{out}.
\label{robust_3}
\end{equation}
From~\eqref{robust_2},~\eqref{robust_3} on varying~$\delta$, a tradeoff occurs between $n_{out}$ and the total number of $\delta$-typical sequences. However, for the choice $\delta \triangleq \delta' n^{\frac{-1}{2}(1-\gamma)}$, where $\delta'>0$ and $0<\gamma<1$, $n_{out}=  \delta' n^{\frac{1}{2}(1+\gamma)} $ is a vanishing fraction of~$n$, and the probability of non-typical sequences would be upper-bounded by $2 |\ss| e^{-\frac{\mu n^{\gamma}}{3}}$. Hence, negligible $n_{out}$ can be guaranteed for almost all sequences satisfying a distribution at large~$n$.

Now we are ready to present the capacity achieving scheme. The encoding, the choice of~$\Dm$ as well as~$\Um$ are identical to Section~\ref{sec:heuristic}, i.e., $D_{ii}^2$ are the normalized optimal power allocations for~$h_i$ and~$U_{ii}$ are the MMSE coefficients. However, designing the permutation matrix~$\Vm$ is now more challenging, since the number of occurrences of the different channel values does not exactly fit the statistical distribution of the channel. We adopt a best-effort approach to choose~$\Vm$, whose design is provided in detail in Table~\ref{table_1}. Briefly, we make a rough assumption that coefficient~$\alpha_k$ occurs exactly~$n \prob_k$ times. Designing the first $ n - n_{out} $ rows of~$\Vm$ is identical to Appendix~\ref{appendix:V}. For the remaining $n_{out}$ rows, one or more slots dedicated for $\alpha_k$ may be exhausted. Hence, we utilize the unoccupied slots dedicated for $\alpha_j$ where $j<k$, whose channel magnitudes are smaller. If all are occupied, we utilize the last $n_{out}$ time slots available. This implies that the $n_{out}$ channel coefficients with the largest magnitudes have no dedicated slots.%
\footnote{The structure of~$\Vm$ preserves the power constraint: the precoder design is identical to Appendix~\ref{appendix:V} up to the first $n-n_{out}$~entries. The last $n_{out}$ entries have negligible impact on the power since $\displaystyle{\lim_{n \to \infty}} \frac{n_{out}}{n} = 0$.}
\begin{table}
\centering
\begin{tabular}{|p{\Tablewidth}| @{}}
\hline
\qquad \qquad \qquad \quad \textbf{Design of the permutation matrix~$\Vm$}  \\
\hline
\footnotesize{
\textbf{Set} $\nu_{out}=0$, $\nu_k^{(0)}=0$ for $k=1,\ldots,|\ss|$. 
}
\\
\footnotesize{
\textbf{for $i=1:n$} 
}
\\
\quad 
\footnotesize{
\textbf{Set} $flag=0$.
}
\\
\quad 
\footnotesize{
\textbf{if $h_i== \alpha_k$} 
}
\\
\quad \quad
\footnotesize{ 
\textbf{Set} $\nu_k^{(i)}=\nu_k^{(i-1)}+1$ and $\nu_l^{(i)}=\nu_l^{(i-1)}$ for $l \neq k$.
}
\\
\quad \quad
\footnotesize{
\textbf{if $\nu_k^{(i)} \leq n p_k$} 
}
\\
\quad \quad \quad
\footnotesize{
\textbf{Set} $\Vm(i,:)$ to $\impulse_m^T$, where $m=\nu_k^{(i)}+ n \sum_{l=1}^{k-1} \prob_l$.
}
\\
\quad \quad \quad
\footnotesize{
\textbf{Set} $flag=1$.
}
\\
\quad \quad
\footnotesize{
\textbf{else} 
}
\\
\quad \quad \quad
\footnotesize{
\textbf{for $j=k-1:-1:1$} 
}
\\
\quad \quad \quad \quad
\footnotesize{
\textbf{if $\nu_j^{(i)} \leq n \prob_j$}
}
\\
\quad \quad \quad \quad \quad
\footnotesize{ 
\textbf{Set} $\nu_j^{(i)}=\nu_j^{(i-1)}+1$ and $\nu_l^{(i)}=\nu_l^{(i-1)}$ for $l \neq j$.
}
\\
\quad \quad \quad \quad \quad
\footnotesize{
\textbf{Set} $\Vm(i,:)$ to $\impulse_m^T$, where $m=\nu_j^{(i)}+ n \sum_{l=1}^{j-1} \prob_l$.
}
\\
\quad \quad \quad \quad \quad
\footnotesize{
\textbf{Set} $flag=1$.
}
\\
\quad \quad \quad \quad \quad
\footnotesize{
break;
}
\\
\quad \quad \quad \quad
\footnotesize{
\textbf{end} 
}
\\
\quad \quad \quad
\footnotesize{
\textbf{end} 
}
\\
\quad \quad \quad
\footnotesize{
\textbf{if \textit{flag}==0} 
}
\\
\quad \quad \quad \quad
\footnotesize{
\textbf{Set} $\nu_{out}=\nu_{out}+1$.
}
\\
\quad \quad \quad \quad
\footnotesize{
\textbf{Set} $\Vm(i,:)$ to $\impulse_m^T$, where $m=\nu_{out}+ \lfloor n- n_{out} \rfloor$.
}
\\
\quad \quad \quad
\footnotesize{
\textbf{end} 
}
\\
\quad \quad
\footnotesize{
\textbf{end} 
}
\\
\quad
\footnotesize{
\textbf{end} 
}
\\
\footnotesize{
\textbf{end} 
}
\\
\hline
\end{tabular}
\vspace{2mm}
 \caption{Steps of designing~$\Vm$ for a general fading channel.}
 \label{table_1}
\end{table}
Following the permutation operation, we use an ellipsoidal decision region $\tilde{\genvoronoi}_1$ as follows
\begin{equation}
\tilde{\genvoronoi}_1 \triangleq \{ \sv \in \real^n~:~ \sv^T \boldsymbol{\tilde{\cov}}^{-1} \sv \leq (1+\epsilon) n \},
\label{voronoi_gen}
\end{equation}
where $\boldsymbol{\tilde{\cov}}$ is a diagonal matrix given by
\begin{equation}
\tilde{\cov}_{ii} =
\begin{cases}
\frac{\Px}{\Px h_{\order(i)}^2 + 1} ~  & \text{for}~i \in \{1,\ldots,  n- n_{out}  \}  \\
\Px & \text{for}~i \in \{ n- n_{out} +1, \ldots,n  \}.
\end{cases}
\label{z_cov_gen}
\end{equation}
Owing to~$\Vm$, $\genvoronoi^{(p)} \subseteq \tilde{\genvoronoi}$, where $\genvoronoi^{(p)}$ is an ellipsoid parametrized by a diagonal auto-correlation matrix~$\boldsymbol{\cov^{(p)}}$ whose elements are in descending order (recall the channel coefficients are in ascending order), and hence achieves capacity. $\genvoronoi^{(p)} \subseteq \tilde{\genvoronoi}$ follows since $\cov_{ii}^{(p)} \leq \tilde{\cov}_{ii}$ for any~$i$, as guaranteed by the structure of~$\Vm$. Then the achievable rate is
\begin{equation}
R < \frac{1}{2n} \sum_{i=1}^{n}  \log{\big( 1 +  h_i^2 \Px^*(h_i) \big)},
\label{RateCSIRn}
\end{equation}
which converges to~$\frac{1}{2} \Ex \Big[ \log{\big( 1 +  h^2 \Px^*(h) \big)} \Big]$, by the weak law of large numbers.
The final step is to show that the suboptimal decision region~$\tilde{\genvoronoi}_1$ has negligible impact on the achievable rate. From~\eqref{lattice_rate} as well as the error analysis in Section~\ref{sec:heuristic}, the gap~$\gap \triangleq C - R$ is bounded by
\begin{align}
\gap = & \, \frac{1}{n} \Big( \log \Big( \frac{\vol(\tilde{\genvoronoi}_1)}{\vol(\genvoronoi^{(p)})} \Big) + o(n) \Big) \nonumber \\
= &  \, \frac{n_{out}}{n} \, \frac{1}{n_{out}} \sum_{i=n-n_{out}+1}^{n} \log (1 + \Px^*_{\order(i)} h_{\order(i)}^2 ) + \frac{o(n)}{n}  \nonumber \\
< &  \, \frac{n_{out}}{n} \, \log (1 + \Px^*_{\order(n)} h_{\order(n)}^2 ) + \frac{o(n)}{n} \,,
\label{gap}
\end{align}
where $o(n)$ satisfies $\displaystyle{\lim_{n \to \infty}} \frac{o(n)}{n}=0$. Hence, $\gap$ vanishes since $\frac{n_{out}}{n} \to 0$ as $n \to \infty$.


\bibliographystyle{IEEEtran}
\bibliography{IEEEabrv,References}

\end{document}